\documentclass[conference]{IEEEtran}
\IEEEoverridecommandlockouts
\usepackage[left=0.625in,top=0.75in,right=0.625in,bottom=1.04in]{geometry}
\newcommand{\fN}{\overrightarrow{\mathcal W}}
\newcommand{\bN}{\overleftarrow{\mathcal W}}
\newcommand{\w}{\overleftrightarrow{\mathcal W}}
\newcommand{\err}{\mathrm{err}}

\newcommand{\floor}[1]{\lfloor #1 \rfloor}

\usepackage{cite}
\usepackage{amsmath,amssymb,amsfonts, amsthm}
\usepackage{algorithmic}
\usepackage{graphicx}
\usepackage{textcomp}
\usepackage{xcolor}
\usepackage{graphicx,ragged2e}
\usepackage{physics}
\usepackage{float}
\usepackage{subfiles}
\usepackage{bbm}
\usepackage{adjustbox} 
\usepackage{dsfont}
\usepackage{xcolor} 
\usepackage[rightcaption]{sidecap}
\usepackage{glossaries}
\usepackage{wrapfig}
\newtheorem{theorem}{Theorem}
\newtheorem{definition}{Definition}
\newtheorem{lemma}{Lemma}

\newacronym{blbc}{BLBC}{Bidirectional Lossy Bosonic Channel}
\newacronym{idr}{IDR}{Intrusion Detection Rate}
\newacronym{dtr}{DTR}{Data Transmission Rate}
\newacronym{otdr}{OTDR}{Optical Time Domain Reflectometer}
\newacronym{lcst}{LCST}{Pure-Loss Coherent State Transmitter}
\newacronym{jcas}{JCAS}{Joint Communication and Sensing}
\def\BibTeX{{\rm B\kern-.05em{\sc i\kern-.025em b}\kern-.08em
    T\kern-.1667em\lower.7ex\hbox{E}\kern-.125emX}}

\begin{document}

\title{Joint Communication and Sensing over the Lossy Bosonic Quantum Channel
\thanks{This work was financed by the DFG via grants NO 1129/2-1 and NO 1129/4-1 and by the Federal Ministry of Education and Research of Germany in the programme of ``Souver\"an. Digital. Vernetzt.''. Joint project 6G-life, project identification number: 16KISK002 and via grants 16KISQ077 and 16KISQ168 and by the Bavarian Ministry for Science and Arts (StMWK) in the project NeQuS (Networked Quantum Systems). Further the support of the Munich Quantum Valley (MQV) which is supported by the Bavarian state government with funds from the Hightech Agenda Bayern Plus and Munich Center for Quantum Science and Technology is gratefully acknowledged.}
}

\author{
\IEEEauthorblockN{Pere Munar-Vallespir, Janis Nötzel}
\IEEEauthorblockA{\textit{Emmy Noether Group Theoretical Quantum System Design} \\
\textit{Technische Universität München}\\
Munich \\
\{pere.munar,janis.noetzel\}@tum.de}
}

\maketitle

\begin{abstract}
    We study the problem of Joint Communication and Sensing for data transmission systems using optimal quantum instruments in order to transmit data and, at the same time, estimate environmental parameters. In particular we consider the specific but at the same time generic case of a noiseless bosonic classical-quantum channel where part of the transmitted light is reflected back to the transmitter. While sending messages to the receiver, the transmitter tries at the same time to estimate the reflectivity of the channel.
    Extending earlier results on similar but finite-dimensional systems, we are able to characterize optimal tradeoffs between communication and detection rates. We also compare quantum performance to analogous classical models, quantifying the quantum advantage.  
\end{abstract}

\section{Introduction}
\gls{jcas} has become a topic of interest in the last decade, with recent activities targeting the investigation of its use for 6G development \cite{jcassurvey}. The goal is to use existing networks to perform sensing in what are called perceptive mobile networks (PMN) \cite{jcassurvey2}. This has motivated research of information-theoretic models that combine both tasks to adequately describe the tradeoff between them. Intuitively, it is expected that a system optimized for data transmission would be worse at sensing than one optimized for sensing (and vice versa), so that nontrivial tradeoffs would arise. For an introduction in the topic see \cite{jcasbook}. As quantum sensing \cite{Dege2017,Pirandola2018} and quantum communication \cite{Gisin2007}  have become fruitful fields, it seems natural to extend \gls{jcas} into the formalism of quantum mechanics. In particular, the superadditivity of quantum capacity and the ability to perform collective measurements that outperform any local ones for state discrimination are a strong motivation for generalizing \gls{jcas} to the quantum regime. Moreover, \gls{jcas} would benefit IOT by allowing each device to not only connect to others but also work as a sensor and then share information with other devices. The intensity of reflections when operating in this manner is expected to be small, thus providing an interesting scenario to investigate quantum advantage, which is known to be larger in low photon regimes. 

Despite its potential advantages and interest, literature on the topic is scarce. A first result in the quantum aspects of transmitter-based \gls{jcas} was studied in \cite{Wang_2022} for finite-dimensional systems. Namely, it was shown how to quantify the performance of a sensing task for a classical-quantum (c-q) channel in open-loop coding (constant code, not adaptive based on measurement of previous measurements) and that in that scenario the performance of codewords only depends on the empirical distribution of symbols in the codewords. In \cite{liu2024quantum}, a protocol was theoretically proposed to use network to perform distributed sensing beyond the Heisenberg limit along the network while performing quantum secure direct communication. In \cite{xu2024integrated}, an experimental approach was shown in distributed sensing in a quantum network that could perform CV-QKD. A different line of work that has also been explored is simultaneous decoding and parameter estimation with limited distortion at the receiver \cite{pereg2021communicationquantumchannelsparameter} with the possibility of  side information being available at the transmitter.  

In the following work we extend the results in \cite{Wang_2022} to an infinite-dimensional c-q channel and fully give its detection and communication rate region. We focus on the pure loss bosonic channel, which is known to be a simple yet useful model for both free space and optical fiber communication. We construct a model where the sender receives a weak backscattered signal that is independent of the one sent to the sender. This is a first step in understanding the sensing-communication tradeoff for bosonic systems. The possibility of using mobile networks makes the idea of studying \gls{jcas} in quantum bosonic systems specially interesting, as it is the most promising field for its application. For communication, we draw on the known results of classical capacities of quantum bosonic channels \cite{Giovannetti_2014}. Detection performance is based in the quantum version of the Chernoff bound. For two hypothesis, its achievability was shown by \cite{Audenaert_2007} and it was shown to be optimal by \cite{Nussbaum_2009}. Then it was extended to an arbitrary number of hypothesis in \cite{Li_2016}.

\section{Notation and problem definition}
\subsection{Notation}
We denote as $\mathcal{H}_a$ the separable complex Hilbert with dimension $a$ and its elements are noted with a ket $\ket{\cdot}$, $\mathcal{D}(\mathcal{H}_a)$ the set of density matrices over $\mathcal{H}_a$ and by $\mathcal{L}(\mathcal{H}_a)$ the set of linear operators over $\mathcal{H}_a$. For an operator $X$, $X>0\iff \bra{\psi}X\ket{\psi}>0\forall\ket{\psi}$. In general, a superindex $x^n$ denotes a sequence of elements of length $n$ with elements $x_i$ and a parenthesis and a superindex, e.g. $(x)^n$, is used to denote exponentiation. For a sequence $x^n$, we denote the type of an element $y$ as $N(y|x^n)$, which is defined as $N(y|x^n) = \sum_{j=1}^n \frac{1}{n}I(x_j^n, u)$ where $I(x,y)$ is the indicator function, which is 1 if $x=y$ and 0 otherwise.  We denote as $\{\ket{n}\}_{n=0}^\infty$ the photon number or Fock basis. $P_n$ is the projector on logarithmic local dimension, that is $P_n = \bigotimes_{i=1}^n\sum_{j=0}^{\floor{\log n}} \ket{j}\bra{j}$. The short-hand $\ket{x^n}\forall x^n\in\mathbb C^n$ refers to $\bigotimes_{i=1}^n \ket{x_i}$ where $\ket{x_i}$ are coherent states. A coherent state is defined as $\ket{x_i} = e^{-|x_i|^2/2}\sum_{n=0}^\infty \tfrac{(x_i)^n}{\sqrt{n!}}\ket{n}$. The Holevo information of a channel $\mathcal W\in C(\mathcal X,\mathcal K)$ and ensemble $\{p_x,\rho_x\}_{x\in\mathcal X}$ is written as $\chi(p;\mathcal W)$. In general, it describes the rate of a classical-quantum (c-q) channel ; thus, the capacity can be calculated as  $C(\mathcal{W}) = \max_{\rho_x, p_x} \chi(p;\mathcal W)$. For a lossy bosonic channel with input energy $E$ and loss $\nu$, its capacity is given by $g(\nu E)$ where $g(x):= (x+1)\log(x+1)+x\log(x)$ is the Gordon function. We denote by $\mathcal D (\rho,\sigma) := \sup_{0\leq s\leq 1}-\log \Tr{(\rho)^s(\sigma)^{1-s}}$ the Chernoff exponent. The fidelity is defined as $\mathcal{F}(\rho, \sigma) := \left(\Tr{\sqrt{\sqrt{\rho}\sigma\sqrt{\rho}}}\right)^2$ which if one of the states is pure can be simplified to $\mathcal{F}(\rho, \ket{\psi}\bra{\psi}) = \bra{\psi}\rho\ket{\psi}$. The trace distance $\mathcal{T}(\rho,\sigma) =\tfrac{1}{2} \Tr{\sqrt{(\rho-\sigma)^\dagger)(\rho-\sigma)}}$. We denote the Total Variation norm for two measures as $||p-q||_{TV}:=\sup_{A}\{|p(A)-q(A)|: A\in\ \Sigma\}$ where the supremum runs over all elements of the $\sigma$-algebra. $\log$ denotes the logarithm in base 2 and $\ln$ the natural logarithm. $\mathcal P(X)$ is the set of probability distributions on set $X$. We denote by $X\sim \mathcal{N}(\mu, \sigma^2)$ that $X$ is normally distributed with expectation value $\mu$ and variance $\sigma^2$ and $Z\sim\mathcal{N}^{\mathbb C}(\mu, \sigma^2)$ means that $\Re{Z}\sim \mathcal{N}(\Re{\mu}, \sigma^2/2)$ and $\Im{Z}\sim\mathcal{N}(\Im{\mu},\sigma^2/2)$ with $\Re{Z}$ and $\Im{Z}$ being independent. We denote by $\mathds 1_n$ the identity operator of dimension $n$. $A^\dagger$ is the hermitian conjugate of a vector or operator $A$. We denote by $\mathbbm i$ the imaginary unit. If no integration bounds are specified, the integral is assumed to be over the entire domain of the function. 
    
\subsection{Problem definition}
\begin{definition}[Bidirectional Lossy Bosonic channel with coherent state input]
We define the \gls{blbc} $\w_{\nu,k}(\cdot)$ as a classical-quantum channel, $\w_{\nu,k}:\mathbb C \to\mathcal{H_\infty}\otimes\mathcal{H_\infty}$, such that $\w_{\nu, k}(\alpha) := \ket{k\alpha}\otimes\ket{\nu\alpha}$. We define the forward and backwards channels $\fN_{\nu,k}(\alpha) = \ket{\nu\alpha}$ and $\bN_{\nu,k}(\alpha) = \ket{k\alpha}$. For a sequence $s^n\in\mathbb C^n$ with elements $s_j^n, j\in [1,n]$, we define the short-hand notation $\w_{\nu,k}(s^n) := \bigotimes_{j=1}^\infty \w_{\nu,k}(s_j)$ and similarly for $\fN(\nu,k)$ and $\bN(\nu,k)$.
\end{definition}
In the following we define the concept of a family of channels and a notion of codes for a family of channels. 
\begin{definition}[$\mathcal{K}$-family of \gls{blbc}s]
    A family of \gls{blbc} is a set of channels $\{\w(\nu, k)\}_{k\in\mathcal{K}}$ for some $\nu\in [0,1]$ and some set finite set $\mathcal{K}$ where $\forall k\in\mathcal{K}, k\in(0,1]$.
\end{definition}
\begin{definition}[($n$,$\epsilon$, $\delta$, $R$, $E$)-code for a $\mathcal{K}$-family of \gls{blbc}s]
An ($n$, $\epsilon$, $D$, $R$, $E$)-code for a family of \gls{blbc}s, is a set comprised by a codebook $\mathcal{C}$, a decoding POVM $\{\Lambda_k\}_{j=1}^M$ and a detection POVM $\{\Pi_k\}_{k\in\mathcal{K}}$ where $\mathcal{C}$ is a set of codewords, $x_1^n, \hdots, x_M^n\in \mathbb{C}^n$ with elements $x_{m,j},j\in[1,n],m\in[1,M]$ where $R = \frac{\log M}{n}$ such that the decoding POVM $\{\Lambda_j\}_{j=1}^n$ on $\mathcal{H}_\infty$ fulfills
\begin{equation}
    \err_{com}(\mathcal{C}) = 1 - \frac{1}{M}\sum_{m=1}^M \Tr{\Lambda_m\fN_{\nu,k} (x^n_m)}\leq \epsilon
\end{equation}
, the sensing POVM $\{\Pi_{k,m}\}_{k\in\mathcal{K}}$ fulfills 
\begin{align}
\label{eq : error-criterion}
    \err_{sen}(\mathcal{C}) &= 1 - \frac{1}{|\mathcal{K}|}\sum_{k\in\mathcal{K}} \Tr{\Pi_{k,m}\bN_{\nu,k} (x^n_m)}\leq \delta \\
    &\forall m \in [1,M] \nonumber
\end{align}
and the codewords fulfill
\begin{equation}\label{eqn:power-constraint}
    \frac{1}{M}\sum_{m=1}^M \sum_{j=1}^n |x_{m,j}|^2 \leq n\cdot E
\end{equation}
\end{definition}
\begin{definition}[Achievable communication-detection tuple $(R,D)$ at energy $E$]
We say a communication-detection tuple $(R,D)$ is achievable for a $\mathcal{K}$-family of \gls{blbc}s with input energy $E$ if there exists a sequence of codes,where the $n$th element is an $(n,\epsilon_n, \delta_n, R_n,E)$-code such that the following conditions are fulfilled:
\begin{align}
    &\lim_{n\xrightarrow[]{}\infty} \epsilon_n =0 \\
    &\lim_{n\xrightarrow[]{}\infty} \inf R_n = R \\
    &\lim_{n\xrightarrow[]{}\infty}\inf -\frac{1}{n}\ln{\delta_n} = D. 
\end{align}
 
The achievable communication-detection region $\mathcal{R}$ is the convex set of all achievable tuples $(R,D)$.
\end{definition}

\section{Main results}
\begin{theorem}[Asymptotic discrimination of sequences of coherent states]\label{thm:discrimination-exponent}
 For a set of sequences of coherent states $\{\ket{k\alpha^n}\}_{k\in\mathcal{K}}$ with a finite amount of different states such that $\mathcal{K}\subset \mathbb C$ and $|\mathcal{K}|<\infty$ and with $f(\beta) := N(\beta|\alpha^n)$, the asymptotic discrimination exponent is given by $D = \min_{k,k',k\neq k'}\sum_{\beta}\frac{f(\beta)|\beta|^2|k-k'|^2}{2}$.
\end{theorem}
\begin{theorem}[Achievable communication detection region of a $\mathcal{K}$-family of \gls{blbc} with coherent states]
\label{theorem2}
The achievable communication-detection region $\mathcal{R}$ is given by all tuples $(a,b)$ such that $0 \leq a \leq g(\nu E)$ and $0\leq b \leq D_{max}$ where $g(\nu E) = (\nu E+1)\log(\nu E +1) - ( \nu E)\log(\nu E)$ and  $D_{max} = \min_{k,k'\in \mathcal{K},k\neq k'}\frac{E|(k-k')|^2}{2}$. This performance is achieved by codes that have constant energy in their codewords. 
\end{theorem}

\section{Comparison with homodyne measurement}
In \cite{Hamdi2022}, a gaussian channel model was solved with two hypothesis in the classical setting and with additive white gaussian noise. This model can be used to describe a pure loss bosonic channel with shot noise limited detection. Shot noise is the minimal variance in the gaussian probability distribution that arises from performing homodyne detection on coherent states due to quantum uncertainty\cite{Banaszek_2020}. Thus, it is the minimal variance of the output measurement and the limit for signal-to-noise ratio. In the following, we use the theorems derived in \cite{Hamdi2022} in the shot noise limited regime to provide a classical comparison to the quantum performance. Note that this classical model assumes single mode homodyne measurements for both outputs, thus it is the performance of a suboptimal yet practically relevant POVM choice. 
\begin{definition}[Classical gaussian bidirectional channel with average white gaussian noise]
\label{def : classical model}
    A classical gaussian bidirectional channel $\mathcal{Q}_{\nu, k}$ is a tuple of classical channels over $\mathbb R$ such that for an input sequence $x^n$, its output sequences $Y^n_A, Y^n_B$ are sequences of random variables which have elements given by $Y_{Ai} = x_i+Z_{Ai}$ and $Y_{Bi} = k\cdot x_i+Z_{Bi}$ where $Z_{Ai}\sim \mathcal{N}(0,\sigma^2_A)$ and $Z_{Bi}\sim \mathcal{N}(0,\sigma^2_B)$, where $k\in\{0,1\}$. We assume the input to have an average energy constraint given by $\tfrac{1}{n}\sum_{i=1}^n |x_i|^2 = E$.
\end{definition}
The reflected part $Y_{Ai}$ can be understood as the output of performing homodyne measurement on $\bN_{\mu, \nu}(x^n)$ and $Y_{Bi}$ is the output at the receiver when performing homodyne measurement on $\fN_{\mu,\nu}(x^n)$ when setting $\nu=1$,$\mathcal{K}=\{0,1\}$ on the \gls{blbc} and setting $\sigma_A^2 = \sigma_B^2 = 1$ on the classical version in Definition \ref{def : classical model}, due to homodyne measurement probability distribution on coherent states being a gaussian with variance 1\cite{homodyne_statistics}. That is, if we consider a noiseless homodyne measurement on both outputs of \gls{blbc} given the input sequence $x^n$,  the outputs of the measurement,  behave like the outputs of channel model in Definition \ref{def : classical model} when sequence $x^n$ is the input.

In the following, we state the result from \cite{Hamdi2022} for this model, which we will use as benchmark for the quantum performance.

\begin{theorem}[\cite{Hamdi2022}]
For the Gaussian channel with AWGN, the achievable region $\mathcal{R}_c$ consists of all tuples $R\leq \tfrac{1}{2}\log(SNR_A)$ and $D\leq \tfrac{1}{8}SNR_B$, where $SNR_A:= E/\sigma^2_A$ and $SNR_B := E/\sigma^2_B$.
\end{theorem}
Note that for detection exponents we take base $e$ instead of base $2$, thus cancelling a factor $\log e$ in the original result. 
For the classical bidirectional lossy channel in the shot noise limit ($\sigma_A = \sigma_B = 1$) with $\mathcal{K}=\{0,1\}$ and $\nu\in [0,1]$ and energy constraint $E$, the achievable region $\mathcal{R}_c$ is thus given by all tuples $(R,D)\in\mathbb R^+$ such that  $R\leq \frac{1}{2}\log(1+\nu E)$ and $D\leq \frac{E}{8}$.
As in the quantum case, there is no tradeoff in the classical version and codes with constant energy are also optimal for both communication and detection. By Theorem \ref{theorem2}, the region of the \gls{blbc} with $|\mathcal{K}|=\{0,1\}$ and $\nu\in[0,1]$ is given by $R\leq g(\nu E)$ and $D\leq \frac{E}{2}$. It is known that the quantum capacity of the lossy bosonic channel has an unbounded advantage over its Shannon (single mode-restricted measurement) capacity \cite{nötzel2022operatingfibernetworksquantum} for low intensity at the receiver. However, for detection this phenomenon does not happen and there a fixed ration of 1/4 between the best quantum measurement and homodyne detection. 
In Figure 1 we plot the two achievable regions, the best achievable one using general quantum strategies $\mathcal{R}$ and the classical one $\mathcal{R}_c$ for low photon number. While detection advantage is independent of photon number, communication advantage $\log(1+E)/g(E)$ scales logarithmically as photon number goes down \cite{nötzel2022operatingfibernetworksquantum}.

\begin{figure}
    \centering
    \includegraphics[width =\linewidth]{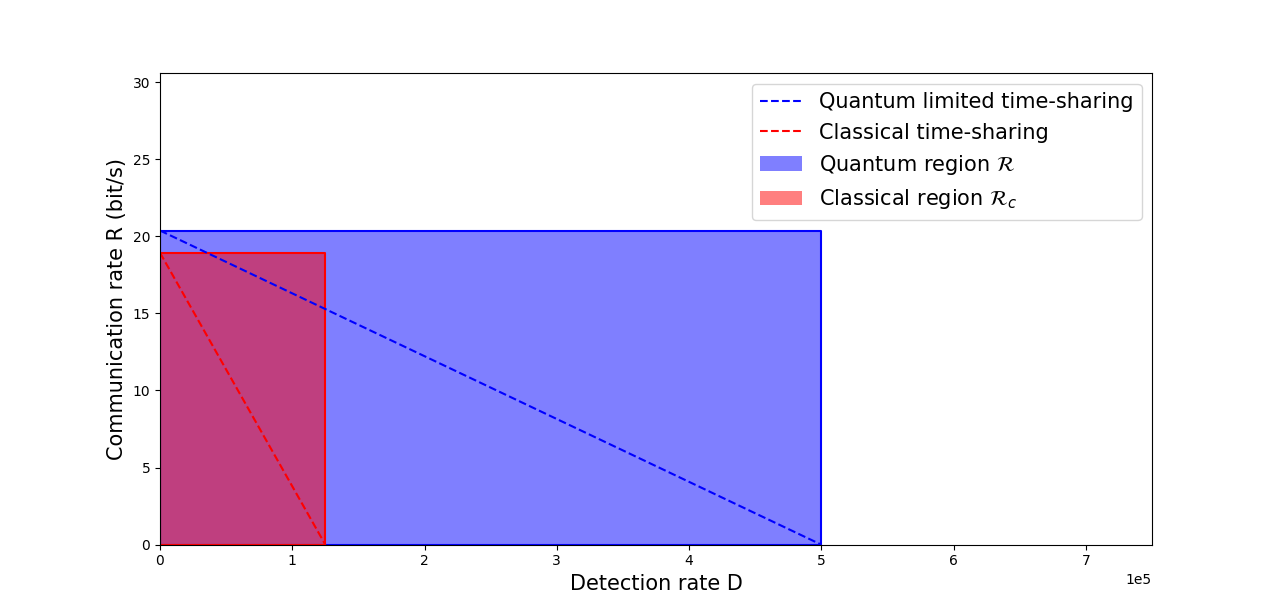}
    \caption{Achievable region for $\mathcal{K}=\{0,1\}$. The red region is the result by \cite{Hamdi2022}, using homodyne measurement on each pulse at sender and receiver. The blue region is our result, achieved using optimal joint quantum measurements. The diagonal dashed lines are given by time-sharing. Using $10^6$ photons per pulse and assuming no loss between sender and receiver. At 1550nm and 10GBd, this corresponds to 1 mW.
}
    \label{fig:enter-label}
\end{figure}
\section{Proofs} 
Our strategy of proof for achievability of detection is to utilize the derivation of \cite{Wang_2022}, which however assumes both finite-dimensional Hilbert spaces and a finite set of possible quantum signal states. Our Lemma \ref{lemma3} ensures that we can approximate the bosonic system under consideration with a finite-dimensional one. 

\subsection{Approximation of sequences of coherent states in finite dimensions}
\begin{lemma}{If $\Tr{P\rho_i}\geq 1-\epsilon\forall i$, for some $\epsilon\in (0,1)$, then  $\Tr{(\bigotimes_{i=1}^n P)\rho^n}\geq 1-n\epsilon$ where $\rho^n = \bigotimes_{i=1}^n\rho_i$}.
\label{lemma1}
\end{lemma}
\begin{proof}
$\Tr{(\bigotimes_{i=1}^n P)\rho^n}=\Pi_{i=1}^n\Tr{P\rho_i}\geq (1-\epsilon)^n\geq 1-n\epsilon$, where the second inequality is given by Bernouilli's inequality. 
\end{proof}
\begin{lemma}{Let $\alpha\in\mathbb C$ and $P_n = \sum_{j=0}^{\floor{\log n}} \ket{j}\bra{j}$. Then  $\Tr P_n\ket{\alpha}\langle\alpha|\geq 1-\epsilon_n$ such that $\lim_{n\rightarrow\infty}n\epsilon_n= 0$.}
\label{lemma2}
\end{lemma}
\begin{proof}
     By defining $N:= \floor{\log(n)}$, the trace can be lower bounded as\cite{Noetzel2024}  
     \begin{align}
         \tr P_n|\alpha\rangle\langle\alpha|&=1-\gamma(N,|\alpha|^2)/N!\\
         &\geq 1-(e\cdot \max\{2,|\alpha|^2\}/N)^N.
     \end{align}
     where $\gamma(a,z) := \int_{0}^z t^{a-1}e^{-t}dt$ is the incomplete Gamma function.  
     Setting $\epsilon_n:=(e\cdot \max\{2,|\alpha|^2\}/N)^N$, it holds 
     \begin{align}
         \lim_{n\to\infty}n\cdot\epsilon_n=0
     \end{align}
     as we wanted to show. 
\end{proof}

\begin{definition}[Projectively approximated code]
    For a sequence of codes $\mathcal{C}^n$ and a sequence of projectors $Q^n$, we denote its projectively approximated sequence of codes as $\mathcal{C}^n(Q^n)$, where codewords in  $\mathcal{C}_l$ are  approximated by
    $\rho_{\alpha_m^l} = Q_l\ket{\alpha_m^l}\bra{\alpha_m^l}Q_l + \ket{0}\bra{0}\Tr{1-Q_l\ket{\alpha_m^l}\bra{\alpha_m^l}Q_l}$. Note that this definition ensures that the resulting approximated states are positive and have trace 1. 
\end{definition}

\begin{lemma}Given a sequence of codes $\mathcal{C}^n$, the projectively approximated code $\mathcal{C}^n(P^n)$ with $P_n=\sum_{j=0}^{\floor{\log n}}|j\rangle\langle j|$ and $P^n:=P_n^{\otimes n}$ converges to $\mathcal{C}^n$. That is, given any $\delta>0$ there exists some $n_0$ such that $\mathcal{F}(\ket{\alpha^l_m}, \rho_{\alpha_m^l})\geq 1 - \delta\forall l\geq n_0$ for all code-words $\ket{\alpha_m^n}$ of $\mathcal{C}^n$.  
    
\label{lemma3}
\end{lemma}
\begin{proof}
From Lemma \ref{lemma1} we know that if $\Tr P_n \ket{\alpha_j}\bra{\alpha_j}\geq 1-\epsilon_n$ holds for all $j=1,\ldots,M$, then $\Tr{ P^n \ket{\alpha_j^n}\bra{\alpha_j^n}}\geq 1-n\epsilon$ for all $j=1,\ldots,M$. By Lemma  \ref{lemma2}, $\epsilon_n$ decreases faster than $1/n$ and thus $\lim_{n\rightarrow \infty}\Tr\bigotimes_{j=1}^n P_j \ket{\alpha_j}\bra{\alpha_j}=1$. Thus, by construction of $\rho_{\alpha_m^n}$, $\lim_{n\rightarrow \infty}\mathcal{F}(\rho_{\alpha_m^n}, \ket{\alpha_m^n})=1$ as we wanted to show.
\end{proof}

\textbf{Remark}: Given an absolutely continuous, Riemann-integrable probability measure $\mu$ on $\mathbb C$, we can construct a sequence of finite discrete measures, $\mu^n$, that converge setwise to $\mu$. A sequence of measures is said to converge setwise if 
\begin{equation}
    \lim_{l\rightarrow \infty} \mu_l(A) = \mu(A)
\end{equation}
for every subset $A\in \mathcal{C}$. 

We take any collection $\mathbf X=\{|\alpha_x\rangle\langle\alpha_x|\}_{x\in\Delta}\subset\mathbb C$ with weights $p(x)\geq0$ satisfying $\sum_{x\in\Delta}p(x)=1$ and define a finite discrete measure $p$ via 
\begin{align}
    p(A)=\sum_{\alpha_x\in A}p(x).
\end{align}

Thus, we can arbitrarily approximate any continuous distribution by a finite discrete one in the sense that for any continuous measure $\mu$ there exists a discrete measure $p$ such that $||p-\mu||_{TV}<\epsilon'\forall\epsilon'>0$. For any choice of random codebook, a discrete version that is arbitrarily close to it can be used and thus we restrict to finite measures to construct the codewords. This also implies that each codeword has an arbitrarily large but finite maximum energy per mode if we follow this construction. 

\subsection{Communication}
When using a random codebook,  \cite{winter2016tight} ensures that the data rate $\chi(\nu,\fN_{\nu,k})$ is reached. Since $p$ is arbitrary, we may choose it such that $\|p(A)-\mu_F(A)\|_{TV}$ where $\mu_F$, $F>0$, is defined via $\mu_F(x)=(\pi E)^{-1}e^{-|x|^2/F}$ becomes arbitrarily small. Choosing $F=E-\epsilon$ ensures the power constraint \eqref{eqn:power-constraint} is asymptotically ensured with high probability and using the continuity of $\chi$ quantity then demonstrates that $g(\nu\cdot E)-\epsilon'$ is an achievable data rate for all $\epsilon'>0$. This construction is similar to the one in \cite{holevoBook}.
\subsection{Detection}
\subsubsection{Achievability}
To show achievability we apply \cite[Theorem 2]{Li_2016} and Lemma \ref{lemma3}, following an analogous proof to the one in \cite{Wang_2022}. This allows us to bound the detection rate of a specific codeword once the codebook has been constructed. We assume that codewords are constructed by sampling some finite alphabet $X\subset\mathbb{C}$ such that $|X|<\infty$ with symbol $\alpha_x$ having a probability of $p(x)$. 
\begin{theorem}[\cite{Li_2016}]
    For a set of operators $\rho_1,\hdots,\rho_r$ such that $\rho_j>0,\rho_j\in\mathcal{H}_a, \forall j$ and some $a<\infty$ with spectral decompositions of $\rho_j$ given by $\sum_{i=1}^{T_i}\lambda_{ij}Q_{ij}$ and $T:= \max\{T_1,\hdots, T_r\}$, then there exists a POVM $\{\Pi\}_{i=1}^{r}$ such that 
    \begin{align}
        1- \sum_{j=1}^r\Tr\rho_j\Pi_j\leq 
        f(r,T) \sum_{i<j} \sum_{l,m} \min\{\lambda_{ij},\lambda_{lm}\}\Tr Q_{ij}Q_{lm}\nonumber
    \end{align}
    where $f(r,T)<10(r-1)^2T^2$.
\end{theorem}
We can use this bound on the approximated code $\mathcal{C}^n(P^n)$ noting that due to our error definition we need to add $\frac{1}{|\mathcal{K}|}$ as prior probability. We refer $\mathcal{C}_a^n$ to $\mathcal{C}^n(P^n)$ for ease of notation. 
\begin{align}
    &\err_{sen}(\mathcal{C}^n_a) \leq \max_{m}f(|\mathcal{K}|, T)
    \sum_{k\neq k'}\frac{1}{|\mathcal{K}|^2} \Tr{\rho_{k\alpha^n_m} \rho_{k'\alpha^n_m}}
\end{align}
In the following we drop the subindex $m$, assuming the implicit maximization over codewords as we use minimum error criterion as defined in Definition \ref{eq : error-criterion}. Then, 
 \begin{align}
 &\err_{sen}(\mathcal{C}^n_a)\leq 
         f(|\mathcal{K}|, T)\nonumber{|\mathcal{K}| \choose 2}\cdot\\&\left( \tfrac{1}{|\mathcal{K}|}
        \max_{k'\neq k} \inf_{s\in [0,1]}\Tr{(\rho_{k\alpha^n})^s(\rho_{k'\alpha^n})^{1-s}}\right)
\end{align}
Then, using the abbreviation $\sup_s$ for $\sup_{s\in[0,1]}$ we give a lower bound for $D$ of the approximated code as 
\begin{align}
    &D(\mathcal{C}^n_a) \geq 
    \min_{k'\neq k}\sup_{s}-\log\Tr{(\rho_{k\alpha^n})^s(\rho_{k'\alpha^n})^{1-s}}-\delta\nonumber \\
    &= \min_{k'\neq k}\sup_{s} -\log \prod_{i=1}^n \Tr{(\rho_{k\alpha_i})^s(\rho_{k'\alpha_i})^{1-s}}-\delta
    \\
    &= \min_{k'\neq k}\sum_{\alpha_x\in X}^n N(\alpha_x|\alpha^n)\mathcal{D}(\rho_{k\alpha_x}, \rho_{k'\alpha_x})-\delta
\end{align}
 where in the first step we used that approximated codewords are product states and the trace of product states is the product of the trace and that since the approximated output can be performed in a local Hilbert space of $\floor{\log n}$ and that the alphabet has size $|X|<\infty$, then by the type counting lemma $T\leq (n+1)^{|\mathcal{H}||X|}= (n+1)^{\floor{\log n}|X|}$ and we can conclude that $\frac{\log{f(|\mathcal K|, T)}}{n}+\log{{|\mathcal K|\choose 2}}/n - \min_k \log{\tfrac{1}{|\mathcal{K}|}}/n <\delta$ for any $\delta>0$ for large enough $n$. 

Note that we define $\rho_{k\alpha_x}$ as the approximated local state for input $\alpha_x$ in the c-q channel, omitting the dependence of the approximation in $n$ for simplicity of notation. In the second step we use the properties of logarithm and the definition of the Chernoff distance. 

 Then, we use the Fuchs-van der Graaf inequality, $1-\sqrt{\mathcal{F}(\rho,\sigma)}\leq \mathcal{T}(\rho,\sigma)\leq \sqrt{1-\mathcal{F}(\rho,\sigma)}$ and by Lemma \ref{lemma3}, $\lim_{n\rightarrow\infty}\mathcal{F}(\rho_{k\alpha^n_{m}}, \ket{k\alpha^n_m})=1\forall m$ and thus $\lim_{n\rightarrow\infty}\mathcal{T}(\rho_{\alpha^n_m,k}, \ket{k\alpha^n_m})=0\forall m$. 
Since it is known that $\lim_{n\rightarrow\infty}\mathcal{T}(\rho, \sigma_n)=0$ and $\lim_{n\rightarrow\infty}|\sigma_n|_1=|\rho|_1$ if and only if $\lim_{n\rightarrow\infty}\Tr{\omega(\sigma_n-\rho)}=0 \forall\omega\in \mathcal{D}(\mathcal{H}_\infty)$\cite{Arazy81} then 
 \begin{equation}
     \lim_{n\rightarrow\infty}\mathcal{D}(\rho_{k'\alpha}, \rho_{k\alpha}) =\mathcal{D}(\ket{k\beta}\bra{k\beta}, \ket{k'\beta}\bra{k'\beta}) 
 \end{equation}

 The Chernoff exponent between two gaussian states that are equivalent up to a displacement is given in \cite[Section VII A]{Calsamiglia_2008} as 
\begin{equation}
    \mathcal{D}_{G} = \tfrac{|d|^2}{2}\left(e^{-2r}\cos^2\theta + e^{2r}\sin^2\theta\right)\tanh\tfrac{\beta}{4}
\end{equation}
where $\theta$ is the squeezing angle, $r$ the squeezing parameter, $\beta$ the inverse temperature and $d$ the relative displacement. As we are using coherent states, we set $r,\theta = 0$ and $\beta = \infty$. The Chernoff exponent of two states $\ket{k'\alpha}$ and $\ket{k\alpha}$ is then given by
\begin{equation}
    \mathcal{D}(\ket{k\alpha}\bra{k\alpha}, \ket{k'\alpha}\bra{k'\alpha}) = \tfrac{|k'\alpha-k\alpha|^2}{2}
\end{equation}
and we note that minimization is independent of $\alpha$. In fact, the minimum is given by the two parameters $k$ and $k'$ that minimize $|k-k'|^2$. Moreover, the Chernoff exponent in this case is a continuous and derivable function of $\alpha$. Then for $\mathcal{C}^n_a$ and for large enough $n$, 
\begin{align}
    D(\mathcal{C}_a^n) &\geq \min_k\min_{k'\neq k}\sum_{\alpha_x\in X} N(\alpha_x|\alpha^n)|k-k'|^2\frac{|\alpha_x|^2}{2}\\
    &= \frac{E}{2}\min_k\min_{k'\neq k}|k-k'|^2
\end{align}
where in the last step we used that $\sum_{\alpha_x\in X}N(\alpha_x|\alpha^n)|\alpha_x|^2$ is the average energy of the codeword. Equation (20) proves that the exponent given in Theorem 1 is achievable. 
Finally, we note that the same sequence of measurements that achieves this bound for the approximated code can be used in the original code, and, due to continuity of the trace and matrix product as stated earlier, this bound is also achievable for the original code. By the law of large numbers, as $n\rightarrow\infty$, $\sum_{\alpha_x\in X} N(\alpha_x|\alpha^n)|k-k'|^2\frac{|\alpha_x|^2}{2}\rightarrow \sum_{\alpha_x\in X} p(x)|k-k'|^2\frac{|\alpha_x|^2}{2}$ and thus Theorem \ref{thm:discrimination-exponent} is proven. 

 Recalling that there is an implicit minimization in the codewords, the performance of the code is given by the codeword with the least amount of energy. If all codewords saturate the energy bound, then the code is optimal.

\subsubsection{Converse}
For the converse, the result given in \cite[Section V, B, c]{Wang_2022} applies also to our channel. The bound they get based on binary hypothesis testing can also be applied to coherent states without any modification, since it is known that for two coherent states the Chernoff bound still applies \cite{nussbaum2013attainment, Mosonyi2023-ut} and having a better performance would imply going over the Chernoff bound.

\section{Conclusions and outlook}
In this work we proposed a simple model for Quantum \gls{jcas}, the \gls{blbc}, that can be analitically solved using the current theoretical knowledge in quantum communication and quantum hypothesis testing. In this channel there is no tradeoff for communication and detection, energy being the limiting factor for both tasks. We compared the performance of optimal quantum measurements to homodyne measurement as given in \cite{Hamdi2022}, finding a factor 4 improvement in detection rate while achieving the Holevo capacity, which is strictly larger than Shannon capacity. Moreover, the Holevo capacity has an unbounded advantage in low photon regimes. 

We note that the model presented in this paper only includes loss and doesn't consider other typical effects that happen in free space or optical fiber communication. Other interesting effects to include in the model would be thermal noise, fiber non-linearities and turbulence. We also assume a POVM at the receiver that achieves the Chernoff bound, which is not possible to implement with the current state of the art. Moreover, this POVM is not constructed directly and might be arbitrarily complex to implement. Extensive experimental research would be needed to develop detectors capable of implementing such a POVM. 

Further research should include gaussian channel models and gaussian input states in general, as both quantum communication in gaussian channels and discrimination of gaussian states are well-known topics \cite{Giovannetti_2014,gaussianHypothesisTesting, Weedbrook_2012}. Thus one expects these models to be analytically solvable while at the same time being practically interesting for tasks such as free space or optical fiber communication.

\bibliography{IEEEabrv, bib}
 
\end{document}